\documentclass{article}[10pt]
\usepackage{fullpage}

\usepackage{amsmath,epsfig,amsfonts,amssymb,amsthm,bm,multirow,makecell,caption,soul,csquotes,color,wrapfig}

\usepackage[boxed, ruled, vlined, linesnumbered]{algorithm2e}

\DeclareMathOperator*{\argmax}{arg\,max}

\newtheorem{definition}{Definition}

\newtheorem{theorem}{Theorem}

\usepackage{hyperref}

 \newenvironment{changemargin} [2]{\begin{list}{}{
          \setlength{\topsep}{0pt}\setlength{\leftmargin}{0pt}
          \setlength{\rightmargin}{0pt}
          \setlength{\listparindent}{\parindent}
          \setlength{\itemindent}{\parindent}
          \setlength{\parsep}{0pt plus 1pt}
          \addtolength{\leftmargin}{#1}\addtolength{\rightmargin}{#2}
          }\item }{\end{list}}

 \newenvironment{myitemize} 
   {
     \begin{changemargin}{-3pt}{-0cm}
     \vspace{-10pt}
     \hspace{-5pt}
     \begin{itemize}
     \setlength{\itemsep}{3pt}
   }
   {
     \end{itemize}
     \vspace{-2pt}
     \end{changemargin}
   }

   {
     \begin{changemargin}{-8pt}{-0cm}
     \vspace{-13pt}
     \hspace{5pt}
     \begin{description}
     \setlength{\itemsep}{-1pt}
   }
   {
     \end{description}
     \end{changemargin}
   }

\newcommand{\myref}[1]{\S\,\ref{#1}}

\newcommand{\scs}[1]{{\small {\sc #1}}\xspace}

\begin{document}
%

%
%
%
%


\title{Inspiration or Preparation?\\
Explaining Creativity in Scientific Enterprise}

%
%
%

\author{
Xinyang Zhang$^\dagger$ \quad Dashun Wang$^\ddagger$ \quad Ting Wang$^\dagger$ \\
$^\dagger$Lehigh University,  xizc15@lehigh.edu, inbox.ting@gmail.com\\
$^\ddagger$Northwestern University, dashunwang@gmail.com
}


\maketitle

\begin{displayquote}
{\em It is the function of creative man to perceive and to connect the seemingly unconnected.}
\flushright \hfill William Plomer
\end{displayquote}

\begin{abstract}

Human creativity is the ultimate driving force behind scientific progress. While the building blocks of innovations are often embodied in existing knowledge, it is creativity that blends seemingly disparate ideas. Existing studies have made striding advances in quantifying creativity of scientific publications by investigating their citation relationships. Yet, little is known hitherto about the underlying mechanisms governing scientific creative processes, largely due to that a paper's references, at best, only partially reflect its authors' actual information consumption. This work represents an initial step towards fine-grained understanding of creative processes in scientific enterprise. In specific, using two web-scale longitudinal datasets (120.1 million papers and 53.5 billion web requests spanning 4 years), we directly contrast authors' information consumption behaviors against their knowledge products. We find that, of 59.0\% papers across all scientific fields, 25.7\% of their creativity can be readily explained by information consumed by their authors. Further, by leveraging these findings, we develop a predictive framework that accurately identifies the most critical knowledge to fostering target scientific innovations. We believe that our framework is of fundamental importance to the study of scientific creativity. It promotes strategies to stimulate and potentially automate creative processes, and provides insights towards more effective designs of information recommendation platforms.
\end{abstract}


\section{Introduction}
\label{sec:introduction}

Among the many propulsions behind scientific progress, one stands out for its magnificent, yet intangible force - human creativity~\cite{Koestler:1964:book,Shadish:1994:book}. While the building blocks of innovations are often embedded in existing knowledge, it is creativity that blends seemingly disparate concepts, ideas and theories~\cite{Jones:2008:science,Evans:2011:science}. Indeed, even a theory as revolutionary as Einstein's special relativity essentially reconciles Newtonian mechanics and Maxwell's electromagnetic theory.

Since Plato's time~\cite{plato:book}, numerous studies in psychology, cognitive science and philosophy have offered a plethora of theories to explain human  creativity~\cite{Csikszentmihalyi:1996:book,Gabora:2013:cogsci}. Despite their importance, today we still lack quantitative understanding of such phenomena.
Thanks to prolific scientific publication archives (e.g., \scs{Dblp}\footnote{http://dblp.uni-trier.de},  \scs{PubMed}\footnote{http://www.ncbi.nlm.nih.gov/pubmed/},  \scs{WoS}\footnote{http://wokinfo.com}), we are now equipped with the lens to study creativity in {\em scientific enterprise} with unprecedented precision. Existing studies have focused on quantitatively gauging scientific papers' creativity by investigating their citation relationships~\cite{Uzzi:2013:science,Doboli:2014:arxiv,Kim:2015:arxiv}. These studies offer convincing evidences that a paper's creativity is measurable by examining how it blends originally disconnected knowledge. Yet, little is known hitherto about the underlying mechanism governing this creative process.

This work represents an initial step towards fine-grained understanding of creativity in scientific enterprise. We argue that, to understand the mechanism underlying scientific creative processes, solely relying on papers' citation relationships is insufficient. After all, a paper's references, at best, only partially reflect its authors' actual information consumption behaviors. First, the references may not include the most critical literature. Second, the references may not provide a holistic view of all the literature inspiring the authors. Finally, to understand the correlation between information consumption and knowledge production, it is imperative to characterize their temporal dynamics; however, the references alone do not indicate when the cited literature were actually consumed by the authors. All these limitations highlight a fundamental gap between reality and perception in current studies.

We overcome these limitations by directly contrasting authors' information consumption behaviors with their knowledge products. In specific, using two web-scale, longitudinal datasets, Microsoft Academic Graph (120.1 million papers across all scientific fields) and Indiana University Click (53.5 billion web requests spanning over 4 years), we conduct a systematic study on creative processes in scientific enterprise. Even though varied privacy and technology constraints preclude the possibility of tracking information consumption and knowledge production at an individual level, by studying their correlation at an organization level, we find remarkable predictability in scientific creative processes: of 59.0\% papers across all scientific fields, 25.7\% of their creativity can be readily explained by information consumed by their potential authors.

Moreover, leveraging these findings, we develop a predictive framework that captures the impact of authors' information consumption over their future knowledge products. Using the aforementioned datasets as an exemplary case, we demonstrate that our framework is able to accurately identify the most critical knowledge to fostering target scientific innovations.

To our best knowledge, this work is among the first to study scientific creative processes within the context of information consumption. We believe that the proposed creativity metrics, in synergy with existing measures (e.g., citation counts), lead to more comprehensive understanding of scientific publications' merits. We also believe that our framework is of fundamental importance to the study of human creativity in general. Foreseeably, it promotes strategies to stimulate and potentially automate creative processes, and provides significant insights towards more effective designs of information recommendation platforms.

The remainder of the paper proceeds as follows. \myref{sec:literature} surveys relevant literature; \myref{sec:data} describes the datasets used in our study; \myref{sec:measurement} presents a general creativity definition that subsumes existing ones; \myref{sec:explanation} explores the predictability in scientific creative processes within the context of information consumption; \myref{sec:model} details our prediction framework and develops efficient optimization algorithms; \myref{sec:measurement},~\myref{sec:explanation} and~\myref{sec:model} all conclude with empirical studies of the proposed models and algorithms; the paper is summarized in~\myref{sec:conclusion}.

\section{Related Work}
\label{sec:literature}

In this section, we review four categories of related work: assessment of creativity, scientific impact prediction, map of science, and computational creativity.

Despite numerous qualitative studies on creativity pertaining to various disciplines: psychology, cognitive science, economics and philosophy~\cite{Koestler:1964:book,Csikszentmihalyi:1996:book,Collins:1998:book,Weitzman:1998:qje,Fleming:2001:ms,Gabora:2013:cogsci}, it is only after the proliferation of scientific publication archives that it is feasible to quantitatively study creativity in scientific enterprise. One active line of inquiry is to develop meaningful creativity metrics. Uzzi et al.~\cite{Uzzi:2013:science} proposed to measure a paper's creativity as atypical pairwise combinations of its referenced work; Fleming~\cite{Fleming:2001:ms} proposed to gauge a patent's novelty using new combinations of patents in its references. This line of work offers convincing evidence that scientific work's creativity is measurable by investigating how it blends originally disconnected knowledge. However, little is known hitherto about the mechanism that triggers such connections. To our best knowledge, this work is among the first to directly bridge this gap by studying creativity within the context of information consumption.

Meanwhile, another line of work has focused on  predicting a paper's long-term impact (primarily measured by its citation count) in its early stages~\cite{Wang:2013:science,Cheng:2014:www,Dong:2015:wsdm,Li:2015:kdd} using various semantic features (e.g., author, content, venue). These studies are complementary to this work in that the useful semantic features can be integrated into our framework to train microscopic (e.g., author-, content-, and venue-specific) creativity models.

Another use of papers' reference relationships is to create citation-based {\em maps of science} or {\em knowledge flow maps}~\cite{Rosvall:2008:pnas,Guevara:2016:arxiv}, which help categorize science and understand papers' trans-disciplinary impact. However, these insights do not help explain creativity of individual scientific work.

Finally, this work is related to the broad area of computational creativity~\cite{Wiggins:2006:kbs,Smedt:2014:arxiv}, which focuses on developing artificial intelligence models that exhibit or generate creativity (e.g., problem solving~\cite{Saunders:2001:ccmcd}, visual creativity~\cite{Colton:2008:aaai}, and linguistic creativity~\cite{Veale:2007:cogsci}). In contrast, this work focuses on understanding and modeling creative processes in scientific enterprise. However, incorporating these intelligence models to enhance the predictive power of our framework would be one promising direction.

\section{Data}
\label{sec:data}

Next we describe the datasets used in our empirical study.

\begin{figure}
\centering
\epsfig{width = 120mm, file = 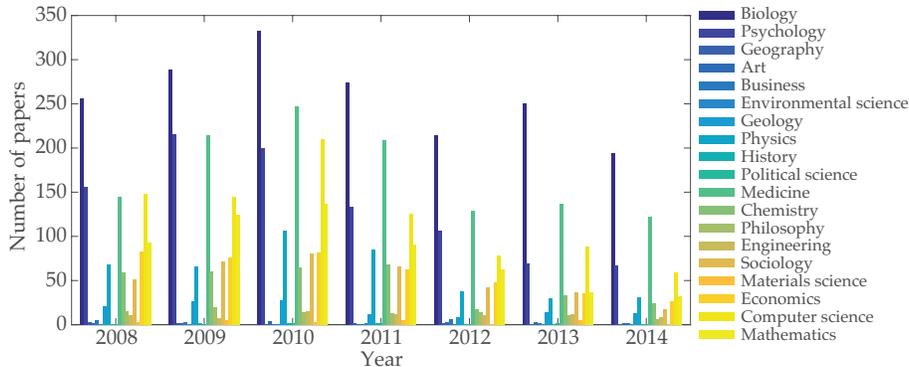}
\caption{Number of Indiana University publications in each scientific field per year. \label{fig:pubperyear}}
\end{figure}

\subsection{Raw Data}

We used two web-scale, longitudinal datasets, corresponding to information consumption and knowledge production of scientific creative processes, respectively.

\subsubsection*{Information Consumption} At present, the most comprehensive dataset that captures information consumption in scientific enterprise is perhaps the web traffic generated by researchers, reflecting how they request and access online resources (e.g., online publication archives). Therefore, in our study, we used the Indiana University Click Dataset~\cite{Meiss:2008:wsdm} (\scs{Click}), which constitutes about 53.5 billion web requests initiated by researchers at Indiana University from 09/2006 to 05/2010. This anonymized dataset was collected by applying a Berkeley Packet Filter to the web traffic passing through the border router of Indiana University and matching all the traffic containing a \scs{HTTP} \scs{GET} request.

Each request consists of the following fields: $\langle${\em timestamp, requested url, referring url, agent, flag}$\rangle$, where ``{\em agent}'' indicates whether the user agent was a browser or a bot, while ``{\em flag}'' indicates whether the request was generated inside or outside of Indiana University. All incoming or bot-generated requests have been filtered.

\subsubsection*{Knowledge Production} Meanwhile, to capture knowledge production in scientific enterprise, we used the Microsoft Academic Graph Dataset~\cite{Sinha:2015:www} (\scs{Mag}). As of 11/06/2015, this dataset consists of 120.9 million papers published in 24,843 venues across all scientific fields.

In a nutshell, the \scs{Mag} dataset is a web-scale entity graph comprising scientific publication records, citation relationships between publications, as well as their authors, author affiliations, publication venues, keywords and fields-of-study (i.e., {\em topics}). In particular, all the topics form a four-level hierarchy, with the highest and lowest levels corresponding to disciplines (e.g., ``{\em Computer science}'') and specific subjects (e.g., ``{\em Decision tree}''), respectively. Each keyword is associated with one topic in the hierarchy.

\begin{figure}
\centering
\epsfig{width = 120mm, file = 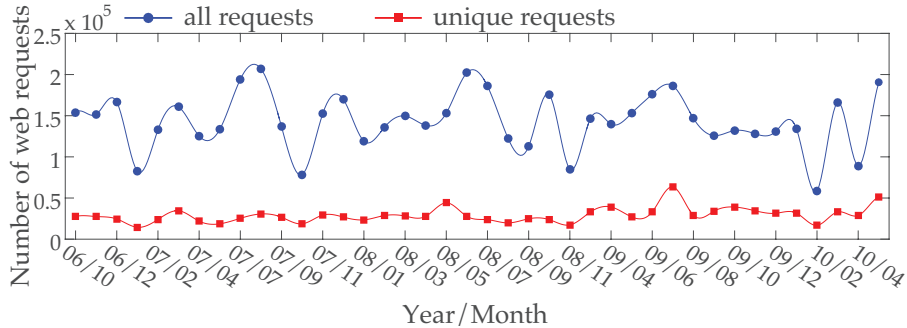}
\caption{Total number of requests and number of unique requests per month. \label{fig:pubpermon}}
\end{figure}

\subsection{Preprocessing}

Next, to match corresponding information consumption and knowledge production data, we correlate the \scs{Click} and \scs{Mag} datasets as follows.

In the \scs{Mag} dataset, we identified all the papers that have at least one author affiliated with Indiana University and were published during the period from 2007 to present, resulting in a collection of 24,399 papers.  Figure~\ref{fig:pubperyear} illustrates the number of papers in each scientific field from 2008 to 2014. It is noted that both the number and composition of papers vary significantly on a yearly basis.

Further, in the
\scs{Click} dataset, we identified all the web requests that ask for papers in the \scs{Mag} dataset by matching \scs{URLs} embedded in the requests with \scs{URLs} of the papers in the \scs{Mag} dataset. The resulting dataset consists of 5.8 million requests for 4.6 million papers (i.e., unique requests). Figure~\ref{fig:pubpermon} illustrates the total number of requests and the number of unique ones per month from 09/2006 to 05/2010. Note that while the total number of requests fluctuates wildly.


\section{Measurement \& Quantification}
\label{sec:measurement}

In this section, we present a general creativity definition and apply it to the aforementioned datasets to empirically study  scientific publications' creativity. We start by introducing a set of fundamental concepts and assumptions used throughout the paper.

\subsection{Preliminaries}

We refer to a scientific publication as a ``paper''. We assume that each paper $k$ is described by a tuple $\langle t^p_k, {\cal K}_k, {\cal C}_k \rangle$, representing its publishing time, keywords and references, respectively. Note that this model can be easily extended by including additional information (e.g., abstract, full text and publishing venue).
Further, we refer to each web request with respect to any paper in online publication archives as a ``reading''. We assume that a reading is described by a tuple $\langle k, t^r_k \rangle$, denoting the reading paper and the time of reading respectively. Note that with respect to a given paper $k$ which is requested (``read'') multiple times, we consider the median of its reading time as $t^r_k$.

Let ${\cal P}_{t, t'}$ denote the papers published within the time period from $t$ to $t'$.
In particular, ${\cal P}_{,t}$ represents all the papers published till $t$.
Similarly, denote by ${\cal Q}_{t, t'}$ the set of papers read during the time window from $t$ to $t'$.
Given that (i) for a large number of papers in the \scs{Mag} dataset, only their publishing years are specified and (ii) even with finer grained timestamps, a paper's publication often deviates from when it is actually finished, we thus use ``year'' as the default time granularity in our study. Thus, with a little abuse of notations, we use $t$ to denote both a timestamp and a one-year-long time window. For example, ${\cal P}_{t}$ represents the set of papers published in year $t$.

\begin{figure}
\centering
\epsfig{file=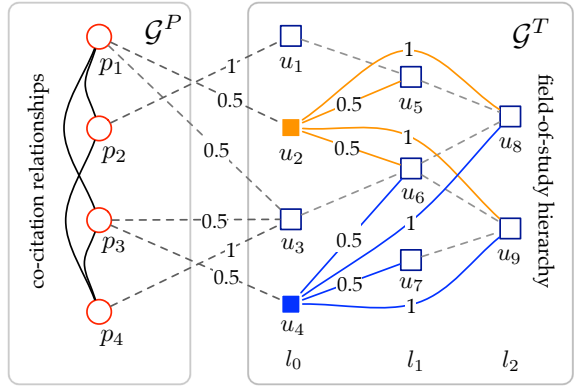, width=80mm}
\caption{A schematic example of paper and field-of-study heterogeneous network. \label{fig:trip}}
\end{figure}

\subsection{A General Creativity Definition}

A variety of creativity metrics have been proposed in literature (e.g.,~\cite{Fleming:2001:ms,Uzzi:2013:science}), all premised on the same intuition: a paper's creativity should be measured by how it blends originally disconnected knowledge. To capture this intuition, one needs to consider two factors: (i) the ``disconnect'' (denoted by $d_{i,j}$) of knowledge represented by two papers $(i,j)$, and (ii) the ``rarity'' (denoted by $r_{i,j}$) that the knowledge of $(i,j)$ has been connected in previously published papers. We define the product of disconnect and rarity,
\begin{equation}
  \label{eq:cs}
\varphi_{i,j} = d_{i,j} \cdot r_{i,j}
\end{equation}
as the creativity score of co-citation of $(i,j)$. We then introduce the following creativity metric.


\begin{definition}[Creativity]
  A paper $k$'s overall creativity $\phi_k$ is the aggregation of creativity scores contributed by all its reference pairs:
\begin{equation}
  \label{eq:cr}
  \phi_k = \ell(\{ \varphi_{i,j} \}_{i, j \in {\cal C}_k})
\end{equation}
  where $\ell(\cdot)$ represents the aggregation function (e.g., average, percentile, maximum).
\end{definition}

We require that $\ell(\cdot)$ is non-decreasing: if $\forall i,j \in {\cal C}_k$, $\varphi_{i,j} \leq \varphi_{i,j}'$, then $\ell(\{ \varphi_{i,j} \}_{i, j \in {\cal C}_k}) \leq \ell(\{ \varphi_{i,j}' \}_{i, j \in {\cal C}_k})$.
We remark that this definition subsumes a number of creativity metrics in literature. For example, the metrics in~\cite{Fleming:2001:ms,Uzzi:2013:science} only consider the rarity factor, while the measures in~\cite{Doboli:2014:arxiv} only reflect the disconnect factor.

\subsection{Rarity versus Disconnect}

Next we discuss the concrete instantiation of $\varphi_{i,j}$. A variety of forms are possible. For example, $r_{i,j}$ can be gauged using the frequency of co-occurrences of $(i,j)$'s publishing venues~\cite{Uzzi:2013:science}. Similar definitions may be given based on papers' author or affiliation information. In our study, we define $(i,j)$'s creativity score within the context of their associated topics.

Recall that each keyword in the \scs{Mag} dataset is associated with one topic. We define paper $k$'s topics (denoted by ${\cal T}_k$) as the aggregation of topics associated with its keyword set $\mathcal{K}_k$. For example, the paper ``{\em Fast algorithms for mining association rules}''~\cite{Agrawal:1994:vldb} is associated with the keyword of ``{\em association rule}'', which is mapped to the topic of ``{\em Association rule learning}''.

Let $\mathcal{P}$, $\mathcal{K}$ and $\mathcal{T}$ be the sets of papers, keywords and topics, respectively.
We introduce a heterogeneous network to describe papers' semantic relationships within the context of their topics, as shown in Figure~\ref{fig:trip}. The paper network ${\cal G}^P = ({\cal P}, {\cal E}^P)$ captures papers' co-citation relationships; each edge $(i,j)\in {\cal E}^p$ indicates that two papers $i,j \in {\cal P}$ are both referenced by some other paper(s). We now introduce the concept of rarity.

\begin{definition}[Rarity]
  The rarity of co-citation of $(i,j)$ till year $t$, $r_{i,j}^t$, is defined as follows:
  \begin{equation}
    r_{i,j}^t = \frac{1}{1 + \log_2(c_{i,j}^t+1)}
  \end{equation}
  where $c_{i,j}^t$ is the number of co-citations of $(i,j)$ till year $t$.
\end{definition}

Meanwhile, the topic network ${\cal G}^T = ({\cal T}, {\cal E}^T)$ encodes the four-level field-of-study hierarchy; a topic $u$ may have one or more parent topics ${\cal H}_u^{l}$ at each higher level $l$. For example, in Figure~\ref{fig:trip}, $u_2$ has two parents $u_5, u_6$ at level $l_1$ and two parents  $u_8, u_9$ at level $l_2$. Given $u' \in {\cal H}_u^{l}$, the weight $w_{u,u'}$ of edge $(u, u') \in {\cal E}^T$ indicates the confidence that $u$ is a sub-topic of $u'$, with $\sum_{u' \in {\cal H}_u^l}
w_{u,u'} =1$.

We define the disconnect of $(i,j)$ based on their topics $({\cal T}_i, {\cal T}_j)$. Let us start with the simplest case that $|{\cal T}_i| = |{\cal T}_j| = 1$, i.e., ${\cal T}_i = \{u\}$ and ${\cal T}_j = \{v\}$. Without loss of generality, we assume that $(u,v)$ lie at the same level of topic hierarchy. The similarity of $(u,v)$ is the aggregation of level-wise similarity of $(u,v)$ and their parents. Three properties are desirable: (i) if $(u, v)$ are identical, their similarity should be 1; (ii) the similarity $s_{u,v}^l$ of $({\cal H}_u^l, {\cal H}_v^l)$ is discounted with respect to $l$; and (iii) $s_{u,v}^l$ is counted only if $\{({\cal H}_u^{l'}, {\cal H}_v^{l'})\}_{l' < l}$ are not identical.

Thus, for $l = 0$ to $3$ (${\cal H}_u^0 = \{u\}, {\cal H}_v^0 = \{v\}$), we compute the level-$l$ similarity as:
$s_{u,v}^l = \sum_{x \in {\cal H}_u^{l} \cap {\cal H}_v^{l}} \min ( w_{u,x},   w_{v,x} )$.
Then the overall similarity of $u,v$ is given as:
\begin{displaymath}
s_{u,v} = \sum_{l = 0}^3  \max \left(1 - \sum_{l' = 0}^{l-1} s_{u,v}^{l'}, 0\right) \cdot s^l_{u, v} \cdot \sigma^l
\end{displaymath}
where the first term represents the ``budget'' after reaching level $l$ and $\sigma$ is the ``discount''.
It can be verified that this definition fulfills all three desirable properties.

For example, in Figure~\ref{fig:trip},  we compute the level-wise similarity of $(u_2, u_4)$:
$s_{u_2, u_4}^0 = 0$, $s_{u_2, u_4}^1 = \min(w_{u_2, u_6}, w_{u_4, u_6}) = 0.5$, $s_{u_2, u_4}^2 = \sum_{u = u_8, u_9} \min(w_{u_2, u}, w_{u_4, u}) = 1$. The overall similarity of $(u_2, u_4)$ is given as:
$s_{u_2, u_4} = 0 + 0.5*0.8 + (1-0.5)*1*0.8^2 = 0.72$,
if $\sigma = 0.8$.

\begin{figure}
\centering
\epsfig{width = 120mm, file = 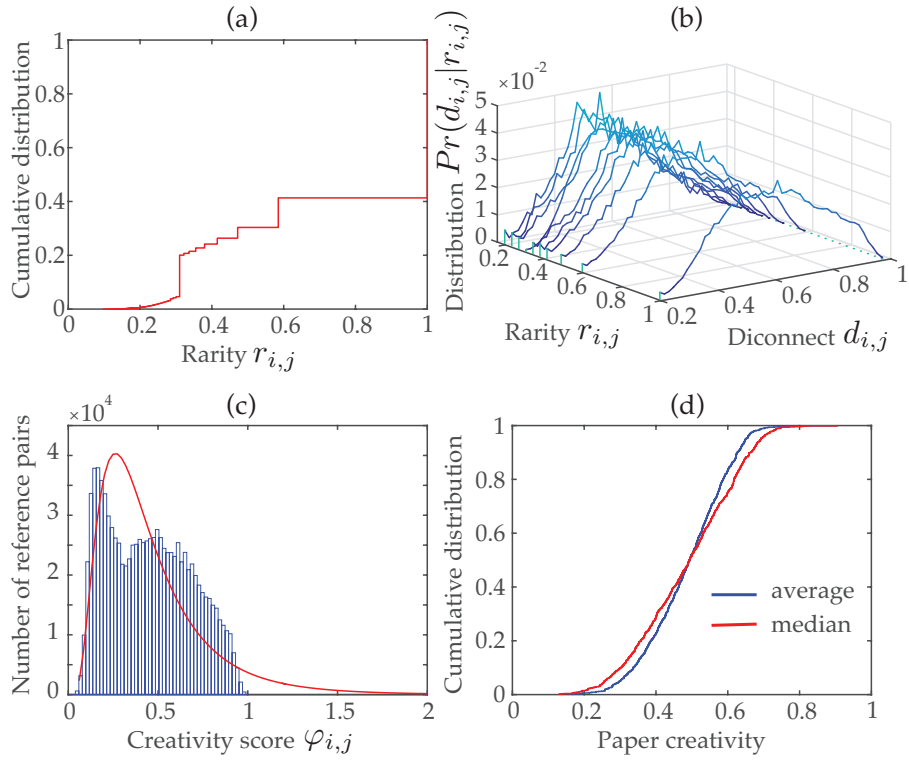}
\caption{(a) Cumulative distribution of reference pairs' rarity; (b) Distribution of reference pairs' disconnect conditional on their rarity; (c) Histogram of reference pairs' creativity scores; (d) Cumulative distribution of papers' creativity (with average and median as aggregation functions). \label{fig:dandr}}
\end{figure}

Now we are ready to introduce the disconnect metric:
\begin{definition}[Disconnect]
The disconnect of two reference papers $(i, j)$ is defined as the average dissimilarity of their topics $({\cal T}_i, {\cal T}_j)$:
\begin{equation}
  \label{eq:un}
d_{i,j} = 1- \frac{1}{|{\cal T}_i||{\cal T}_j|} \sum_{u \in {\cal T}_i, v \in {\cal T}_j} s_{u,v}
\end{equation}
\end{definition}

\subsection{Empirical Study}

Next we apply the metrics above in our empirical study. Due to space limitations, here we only show results obtained using papers published in 2011. Similar phenomena are observed regarding papers in other years.

\subsubsection*{Rarity, Disconnect and Creativity}

With a little abuse of notations, let ${\cal P}$ denote all the papers published in 2011 by Indiana University. We first measure the rarity and disconnect of all the reference pairs of ${\cal P}$, i.e., $ \{ (i,j) \in {\cal C}_k\}_{k \in {\cal P}}$. The results are shown in Figure~\ref{fig:dandr}.

Figure~\ref{fig:dandr}(a) shows the cumulative distribution of $(i,j)$'s rarity $r_{i,j}$. Note that around 60\% pairs have zero co-citation before 2011. In Figure~\ref{fig:dandr}(b), we further measure the conditional distribution of $(i,j)$'s disconnect $d_{i,j}$ on its rarity $r_{i,j}$. The conditional distribution $Pr(d_{i,j}|r_{i,j})$ is fairly similar irrespective to varying $r_{i,j}$, implying that disconnect and rarity are two critical, but complementary factors of creativity. Then, according to Eqn.(\ref{eq:cs}), we integrate $(i,j)$'s rarity and disconnect to compute its creativity score $\varphi_{i,j}$. As shown in Figure~\ref{fig:dandr}(c), the histogram of $\varphi_{i,j}$ roughly follows a lognormal distribution. It can be intuitively explained by that most reference pairs represent modestly ``common'' combinations, while both extremely ``clich\'{e}'' and ``atypical'' combinations are rare. Finally, we measure the creativity of papers in ${\cal P}$ by aggregating their reference pairs' scores. We consider both average and median as the aggregation function $\ell(\cdot)$, with results depicted in
Figure~\ref{fig:dandr}(d). Clearly, most of the papers show moderate creativity, which is consistent with the results reported in prior studies~\cite{Uzzi:2013:science}.

\begin{figure}
\centering
\epsfig{width = 120mm, file = 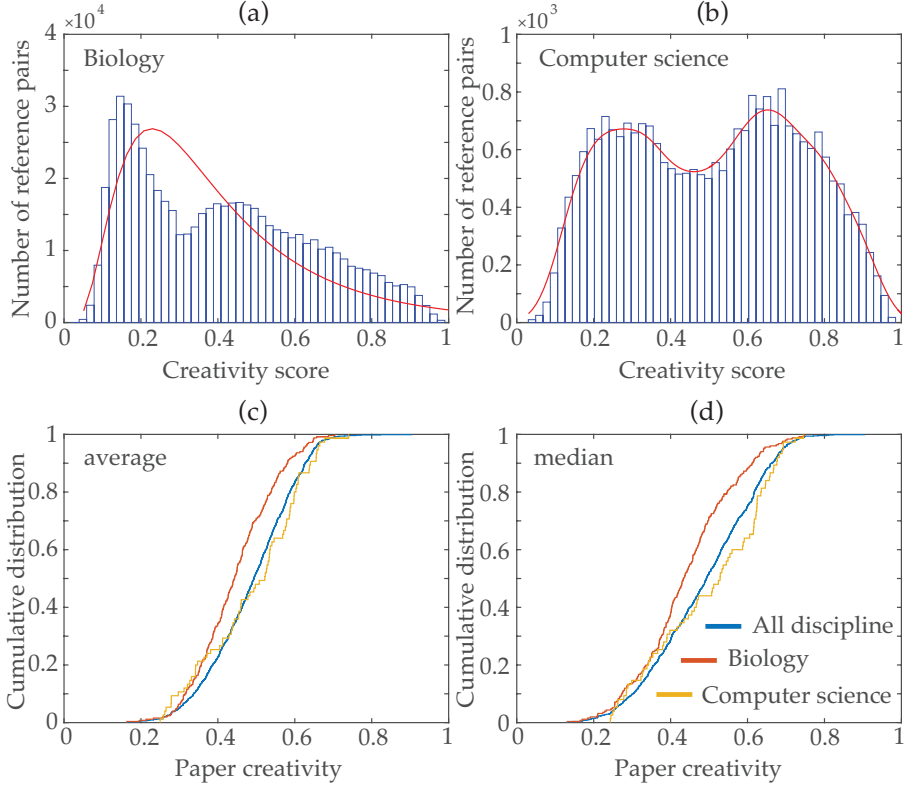}
\caption{(a)-(b) Histogram of reference pairs' creativity scores ({\em Biology} and {\em Computer science}); (c)-(d) Cumulative distribution of papers' creativity (average and median). \label{fig:disp}}
\end{figure}

\subsubsection*{Discipline-Specific Patterns}

We further investigate the discipline-specific patterns of creativity measures. Figure~\ref{fig:disp}(a)-(b) demonstrate the histogram of $(i,j)$'s creativity score $\varphi_{i,j}$ in the disciplines of {\em Biology} and {\em Computer science} respectively. Interestingly, one can observe that $\varphi_{i,j}$ roughly follows a lognormal distribution in Figure~\ref{fig:disp}(a). Meanwhile, in Figure~\ref{fig:disp}(b), $\varphi_{i,j}$ apparently follows a bimodal distribution, peaking at both low and high creativity scores. Such phenomena may be explained by that compared with {\em Biology}, {\em Computer science} is a relatively ``engineering'' discipline, featuring more frequent fusion of originally disconnected knowledge. As shown in Figure~\ref{fig:disp}(c)-(d), this difference indeed leads to the disparate creativity distributions of papers published in these two disciplines: in general, the papers in {\em Computer science} demonstrate higher creativity.

\section{Anatomy of Creativity}
\label{sec:explanation}

Equipped with the metrics introduced above, we are able to quantitatively assess the creativity of a paper, an author, an institution or even a discipline. Next we will explore the {\em predictability} of creativity by directly contrasting authors' raw information consumption behaviors against their knowledge products. We intend to answer the following key question:

\begin{quote}
{\em Which part of a paper's creativity can be explained by papers read by its authors?}
\end{quote}

\hfill

\subsection{Production-Consumption Dependency}
\label{sec:dependency}

We begin with validating the dependency of authors' publishing papers on their reading papers. In specific, for the papers published in year $t$, ${\cal P}_t$, we compare the set of prior work ${\cal C}_t$ referenced by ${\cal P}_t$ (i.e., ${\cal C}_t =  \cup_{k \in {\cal P}_t} {\cal C}_k$) against the set of reading papers ${\cal Q}_{t'}$ in an earlier year $t'$ $(t' < t)$. We measure the dependency of ${\cal P}_t$ on ${\cal Q}_{t'}$ at two distinct levels.

\begin{figure}
\centering
\epsfig{width = 120mm, file = 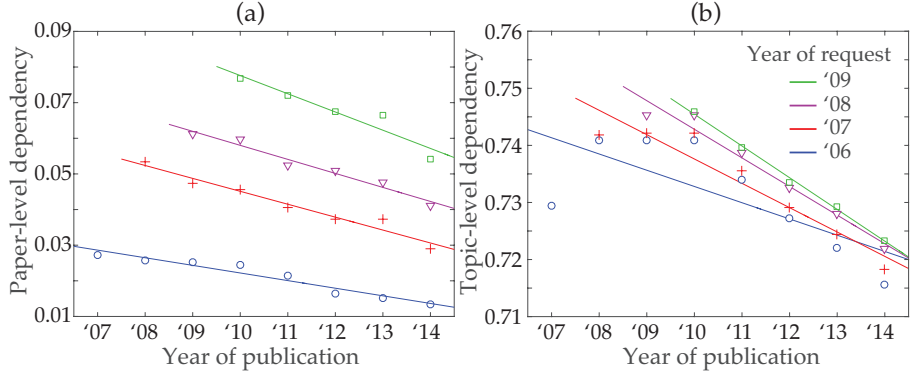}
\caption{Paper- and topic-level dependency of knowledge production on information consumption. \label{fig:dep}}
\end{figure}

\subsubsection*{Paper-Level Dependency}

We first directly compare reference set ${\cal C}_t$ against reading set ${\cal Q}_{t'}$. In specific, we measure the paper-level dependency of ${\cal P}_t$ on ${\cal Q}_{t'}$ by computing the Jaccard's coefficient of ${\cal C}_t$ and ${\cal Q}_{t'}$:
\begin{displaymath}
\frac{|{\cal C}_t  \cap  {\cal Q}_{t'} |}{\min\{  |{\cal C}_t |,  | {\cal Q}_{t'} |  \}}
\end{displaymath}

Figure~\ref{fig:dep}(a) illustrates the paper-level dependency in our datasets. As $t$ varies from 2007 to 2014, we measure the dependency of ${\cal P}_t$ on ${\cal Q}_{t'}$ for $t' (t' < t)$ ranging from 2006 to 2009.

It is observed that this paper-level dependency demonstrates interesting temporal dynamics. First, given reading set ${\cal Q}_{t'}$ (i.e., fixed $t'$), the dependency of ${\cal P}_t$ on ${\cal Q}_{t'}$ decreases as $t$ grows, implying that the impact of ${\cal Q}_{t'}$ gradually decays over time. Second, for given publication set ${\cal P}_t$ (i.e., fixed $t$), its dependency on ${\cal Q}_{t'}$ increases with $t'$, implying that more recent reading papers exert more influence over future publications.

\subsubsection*{Topic-Level Dependency}

We then examine the dependency of publication set ${\cal P}_t$ on reading set ${\cal Q}_{t'}$ at the topic level.

In specific, following the definition of topic disconnect in Eqn.(\ref{eq:un}), we compute the topic-level dependency of ${\cal P}_t$ on ${\cal Q}_{t'}$ as the average ``connectedness'' of papers in ${\cal P}_t$ and ${\cal Q
}_{t'}$:
\begin{displaymath}
1 - \frac{\sum_{i \in {\cal P}_t}\sum_{j \in {\cal Q}_{t'}} d_{i,j}}{|{\cal P}_t||{\cal Q}_{t'}|}
\end{displaymath}

The measurement results are illustrated in Figure~\ref{fig:dep}(b). We have the following observations.

The topic-level dependency shows pattens similar to the paper-level dependency: (i) the impact of reading papers over future publications decays over time; (ii) more recent reading papers exert stronger influence. Nevertheless, compared with the paper-level dependency, the topic-level dependency seems less ``stratified'' in that adjacent years show more resemble patterns. For instance, the dependency with respect to the reading papers in 2006 and 2007 shows strong similarity. Such phenomena can be explained by that authors' interests in particular topics are more stable than their interests in concrete papers.

 \subsection{A Dichotomic Theory of Creativity}
 \label{sec:theory}

The study on paper- and topic-level dependency above offers empirical evidences for our premise that authors' future publications are heavily influenced by prior literature which they have read. This also hints that it is conceivable to answer the {\em creativity prediction} question posed at the beginning of this section.

\subsubsection*{Impact of Information Consumption}

To answer this question, we introduce a mathematical model to quantify the impact of information consumption in creative processes. To motivate the rationale behind our model, let us consider the  following concrete example of three papers by three disjoint groups of authors.

\begin{myitemize}
\item $i$: {\small ``{\em Fast algorithms for mining association rules}''. R. Agrawal and R. Srikant. {\em VLDB '94}}.

\item $x$: {\small ``{\em Frequent subgraph discovery}''. M. Kuramochi and G. Karypis. {\em ICDM '01}}.

\item $j$: {\small ``{\em SPIN: mining maximal frequent subgraphs from graph databases}''. J. Huan, W. Wang, J. Prins and J. Yang. {\em KDD '04}}.
\end{myitemize}

Here paper $j$ differs from $i$ significantly. However, after reading paper $x$, one is probably able to draw the connection from $i$ to $j$ as ``{\em frequent pattern}'' $\rightarrow$ ``{\em frequent subgraph pattern}'' $\rightarrow$ ``{\em maximal frequent subgraph pattern}''. Informally, reading $x$
offers the opportunity to bridge the knowledge gap between $i$ and $j$.

\begin{figure}[h]
\centering
\epsfig{width = 60mm, file = 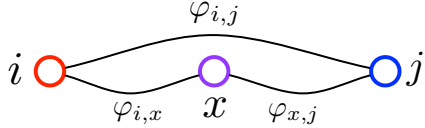}
\caption{Impact of $x$ on the creativity required to connect $(i,j)$: $\Delta^{i,j}_x = \varphi_{i,j} - \min\left(\varphi_{i,j}, \max\left(\varphi_{i,x}, \varphi_{x,j}\right)\right)$. \label{fig:jump}}
\end{figure}

We capture this intuition using the metric of creativity score defined in Eqn.(\ref{eq:cs}). The term $\varphi_{i,j}$ can be interpreted as the ``difficulty'' of connecting $(i,j)$, while the exposure to $x$ may potentially reduce this difficulty. That is, instead of directly connecting $(i, j)$, one may now first connect $i$ to $x$, and then link $x$ to $j$, as illustrated in Figure~\ref{fig:jump}.

Formally, we introduce a metric to describe the impact of paper $x$ on the creativity required to connect $(i,j)$:
\begin{equation}
\Delta^{i,j}_x = \varphi_{i,j} - \min\left(\varphi_{i,j}, \max\left(\varphi_{i,x}, \varphi_{x,j}\right)\right)
\end{equation}
where the term $ \max\left(\varphi_{i,k}, \varphi_{k,j}\right)$ represents the difficulty of connecting $(i,x)$ and $(x, j)$.

Note that this definition can be generalized to the case that $n$ papers ${\bm x} = \{x_1, x_2, \ldots, x_n\}$ collectively bridge the gap of $(i,j)$ by creating an $n$-hop ``path'' between them. Let $\varphi_{\bm x} = \max_{l =1 }^{n-1} \varphi_{x_{l}, x_{l+1}}$. We may define the impact of $\bm{x}$ to $(i,j)$ as:
$\Delta^{i,j}_{\bm x} = \varphi_{i,j} - \min\left(\varphi_{i,j}, \max\left(\varphi_{i,x_1}, \varphi_{\bm x} , \varphi_{x_n, j}\right)\right)$.
Due to space limitations, in the following discussion we focus on the case of $n =1$.

\subsubsection*{Preparation versus Inspiration}

We are now ready to perform an anatomy of the creativity $\varphi_{i,j}$ required to connect the knowledge represented by two papers $(i,j)$. For simplicity, consider the case that during this creative process, the authors read a single paper $x$. We  divide $\varphi_{i,j}$ into two complementary parts.
\begin{myitemize}
\item ``Preparation'' - the marginal reduction in the difficulty of connecting $(i,j)$, due to the authors' reading, which is quantified by $\Delta^{i,j}_x$.
\item
``Inspiration'' - the part of creativity that cannot be explained by the authors' reading, which is quantified by
$(\varphi_{i,j} - \Delta^{i,j}_x)$.
\end{myitemize}

Since a given paper blends the knowledge in all its references, a dichomatic theory naturally follows. That is, a paper's creativity reflects a superimpose of both effects:
\begin{quote}
\boxed{{\bf Creativitiy} = {\bf Preparation} + {\bf Inspiration}}
\end{quote}

Formally, consider a paper $k$ with ${\cal C}_k$ as its reference set. Assume that $k$'s authors have read a set of papers ${\cal Q}$ during their creative process. As given in Eqn.(\ref{eq:cr}), $k$'s overall creativity is assessed by: $\phi_k =\ell(\{\varphi_{i,j}\}_{i,j \in {\cal C}_k})$. For each pair $i, j \in {\cal C}_k$, we identify a reading paper in ${\cal Q}$ that maximally impacts $\varphi_{i,j}$ and compute its impact as:
\begin{displaymath}
\Delta_{\cal Q}^{i,j} = \max_{x \in {\cal Q}}\Delta^{i,j}_x
\end{displaymath}

\begin{definition}[Enabler]
For reference pair $i,j \in {\cal C}_k$, the paper $x^*$ in reading set ${\cal Q}$  that maximally impacts $\varphi_{i,j}$, $x^* = \arg\max_{x \in {\cal Q}}\Delta^{i,j}_x$ is called an enabler.
\end{definition}

Thus, after discounting the preparation quantity embodied in ${\cal Q}$, the inspiration of $k$, $\chi_k$, is measured as:
\begin{displaymath}
{\chi}_k = \ell(\{\varphi_{i,j} - \Delta^{i,j}_{{\cal Q}}\}_{i,j \in {\cal C}_k})
\end{displaymath}
while the preparation quantity $\psi_k$ is computed as the difference of creativity and inspiration measures:
\begin{displaymath}
\psi_k = \ell(\{\varphi_{i,j}\}_{i,j \in {\cal C}_k}) -\ell(\{\varphi_{i,j} - \Delta^{i,j}_{{\cal Q}}\}_{i,j \in {\cal C}_k})
\end{displaymath}

\begin{figure}
\centering
\epsfig{width = 120mm, file = 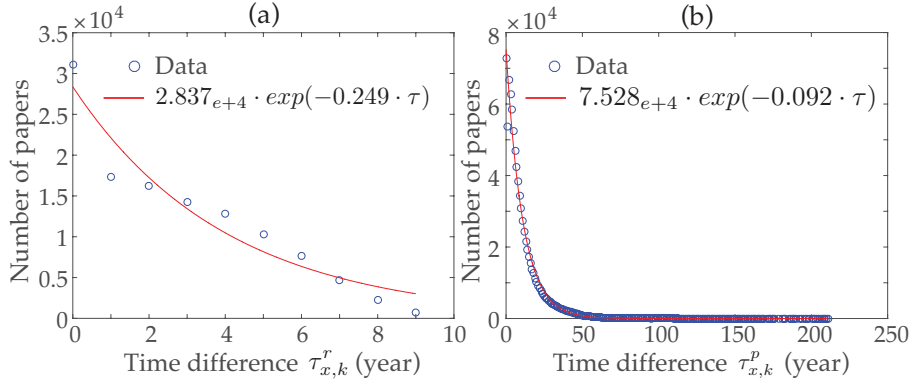}
\caption{Temporal dynamics of information consumption-production. (a) Histogram of consumption-publication temporal interval $\tau_{x,k}^r$; (b) Histogram of publication-publication temporal interval $\tau^p_{x,k}$. \label{fig:correlation}}
\end{figure}

\subsubsection*{Temporal Dynamics}

In the creativity theory above, the definition of ${\cal Q}$ is critical for accurately quantifying preparation and inspiration. While one may regard all the papers read by given authors in the past as ${\cal Q}$, this simplification ignores the temporal dynamics of information production-consumption dependency observed in~\myref{sec:dependency}, resulting in an overestimation of preparation quantities. Here we will address this issue.

It is extremely difficult to directly assess the temporal dynamics of a reading paper $x$'s impact over a future publishing paper $k$. However, by studying in general how authors access papers in the intersection of reference and reading sets (i.e., ${\cal C}_k \cap {\cal Q}$), we are able to build a surrogate temporal dynamics model.

\begin{figure}
\centering
\epsfig{width = 120mm, file = 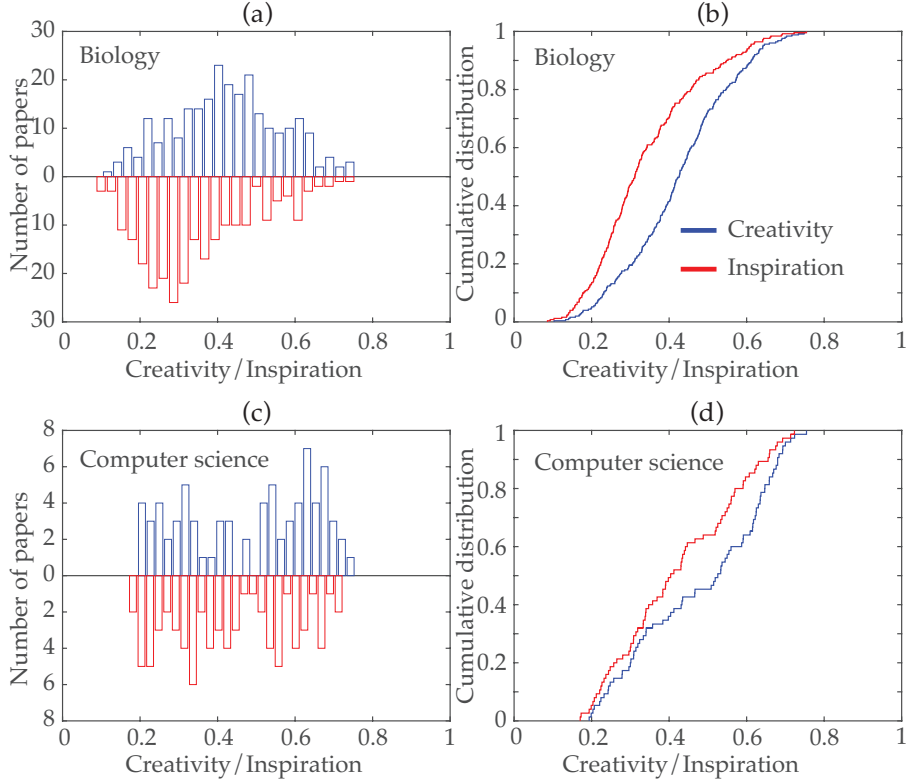}
\caption{(a)-(c) Papers' creativity and inspiration in the disciplines of {\em Biology} and {\em Computer science}; (b)-(d) Cumulative distribution of papers' creativity and inspiration in the disciplines of {\em Biology} and {\em Computer science}. \label{fig:diff}}
\end{figure}

Specifically, for $x \in {\cal C}_k \cap {\cal Q}$, we examine $x$'s reading time ($t^r_x$) and $k$'s publishing time ($t^p_k$). The difference ($\tau^r_{x,k} = t^p_k - t^r_x$) reflects the temporal interval between information consumption and knowledge production. Figure~\ref{fig:correlation}(a) shows that this interval follows a log-log distribution, which is consistent with the studies on papers' ``aging'' phenomena~\cite{Dorogovtsev:2003:lognormal}. Let $Pr(\tau)$ represent this distribution. By discretizing and normalizing $Pr(\tau)$, we compute the probability that a reading paper influences a future paper published $\Delta t$ years later, i.e., $m(\Delta t) = \int_{\tau = \Delta t}^{+\infty} Pr(\tau) \textrm{d}\tau$.

Putting everything together, Algorithm~\ref{alg:sampling} sketches how to compute a given paper $k$'s preparation and inspiration quantities. For each reference pair $i,j \in {\cal C}_k$, the papers in reading set ${\cal Q}$ are ranked according to their impact to the creativity score of $(i,j)$; the current top one $x^*$ is picked with probability $m(\tau^r_{x^*,k})$; if $x^*$ is not picked, it moves to the next most impactful one and repeats this process; finally, $k$'s preparation and inspiration measures are computed by aggregating the creativity scores of all reference pairs.

\begin{algorithm}
  \SetAlgorithmName{Algorithm}{List of operations}
  \SetAlgoLined
  \SetSideCommentRight
  \SetNoFillComment
  \KwIn{publishing paper $k$, reading papers ${\cal Q}$}
  \KwOut{preparation and inspiration of $k$}


  \For{each pair of $i,j \in {\cal C}_k$}{
  \tcc{\small preparation and inspiration to connect $i,j$}
     candidate references $\tilde{{\cal Q}} \leftarrow {\cal Q}$\;
     \While{true}{
     find $x^* = \arg\max_{x \in \tilde{{\cal Q}}} \Delta_{x}^{i,j}$\;
     remove $x^*$ from $\tilde{\cal Q}$\;
     \tcc{\small temporal dynamics-based Bernoulli trial}
     pick $x^*$ with probability $m(\tau_{x^*,k}^r)$\;
     \lIf{$x^*$ is picked}{break}
     }
     compute $\Delta_{\cal Q}^{i,j} \leftarrow \Delta_{x^*}^{i,j}$\;
  }
  \tcc{\small aggregate at paper level}
  $\chi_k$ $\leftarrow$ $\ell(\{\varphi_{i,j} - \Delta^{i,j}_{{\cal Q}}\}_{i,j \in {\cal C}_k})$\;
  $\psi_k$ $\leftarrow$ $\ell(\{\varphi_{i,j}\}_{i,j \in {\cal C}_k}) -\ell(\{\varphi_{i,j} - \Delta^{i,j}_{{\cal Q}}\}_{i,j \in {\cal C}_k})$\;
   \caption{Explaining creativity in scientific work \label{alg:sampling}}
\end{algorithm}

\subsection{Empirical Study}

Next we apply this theory to explain the creativity of papers in our datasets.

\subsubsection*{Inspiration-Preparation Decomposition}

Applying Algorithm~\ref{alg:sampling}, we first assess the preparation and inspiration quantities of papers published in 2011 (similar phenomena are observed in other years). In particular, we examine the papers in the disciplines of {\em Biology} and {\em Computer science} separately.

Figure~\ref{fig:diff} contrasts a paper $k$'s overall creativity $\phi_k$ and inspiration $\chi_k$. Specifically, Figure~\ref{fig:diff}(a)-(c) show the histograms of $\phi_k$ and $\chi_k$ regarding the papers in {\em Biology} and {\em Computer science}, respectively; Figure~\ref{fig:diff}(b)-(d) further compare their cumulative distributions. It is observed in both cases, $\psi_k$ accounts for a sizable portion of $\phi_k$. Table~\ref{tab:anatomy} summarizes the impact of $\psi_k$. We examine the papers whose creativity measures drop after the preparation is taken into account. For example, of about 59.0\% papers across all scientific fields, 25.7\% of their creativity can be explained by information consumed by their authors. Also interestingly, the quantity of preparation varies with specific disciplines: it counts for 21.9\% and 29.2\% of creativity in the disciplines of {\em Computer science} and {\em Biology}, respectively. This variance may be explained by that in a more established discipline as {\em Biology}, authors need to ground their research in existing work more profoundly.

\begin{table}[h]{\small
\centering
\begin{tabular}{c|c|c|c}
\bf{Discipline} & \bf{Papers w. $\downarrow$ (\%)} & \bf{Avg. $\downarrow$} & \bf{Avg. $\downarrow$ (\%)}\\
\hline
\hline
Biology & 64.54\% & 0.132 & 29.18\%\\
\hline
Computer science & 48.00\% & 0.123 & 21.88\%\\
\hline
All disciplines & 58.96\% & 0.133 & 25.67\%\\
\hline
\end{tabular}
\caption{Impact of preparation over overall creativity. \label{tab:anatomy}}}
\end{table}

\subsubsection*{Case Study}

We further explore the practical use of this creativity theory. Here we report one concrete case.

The target paper $k$: {\small ``{\em Engaging online learners: the impact of web-based learning technology on college student engagement}''} cites two papers
$i$: {\small ``{\em A comprehensive look at online student support services for distance learners}''} and $j$: {\small ``{\em Do computers enhance or detract from student learning?}''}. The reading paper $x$: {\small ``{\em The convergent and discriminant validity of NSSE scalelet scores}''} reduces $\varphi_{i,j}$ from 0.661 to 0.495.

At a first glance, $x$ is not relevant to either $i$ or $j$. Yet, after manually examining the full text of these papers, one may notice the following connection: to build an online learning environment ($i$'s knowledge), it is critical to effectively assess student engagement in this new environment ($j$'s knowledge), while $x$ indeed provides a metric to evaluate student engagement.

\subsubsection*{Relationships with Long-Term Impact}

Finally, we investigate the relationships of preparation and inspiration with other important metrics of scientific work. In particular, we focus on a paper's citation count, the de facto metric of its long-term impact~\cite{Wang:2013:science,Cheng:2014:www,Dong:2015:wsdm,Li:2015:kdd}.

Figure~\ref{fig:influence} plots the correlation of a paper $k$'s preparation $\psi_k$ and inspiration $\chi_k$ with its citation count $c_k$. Again, we examine papers in the disciplines of {\em Biology} and {\em Computer science} separately. In both cases, $c_k$ is positivity correlated with $\psi_k$, implying that a paper grounded deeper into exiting work is better received by the research community. More interesting is the correlation of $c_k$ and $\chi_k$: it is positive in the case of {\em Computer science}, yet slightly negative in the case of {\em Biology}. This may be explained by that as a younger and more vibrant discipline, the {\em Computer science} community tends to be more welcoming to radical ideas that blend previously disconnected knowledge.

Clearly, the creativity, preparation and inspiration metrics provide a new perspective to assess scientific publications' merits. We envision that in synergy with other metrics (e.g., citation count), they may lead to more comprehensive understanding of long-term impact of scientific work.

\begin{figure}
\centering
\epsfig{width = 120mm, file = 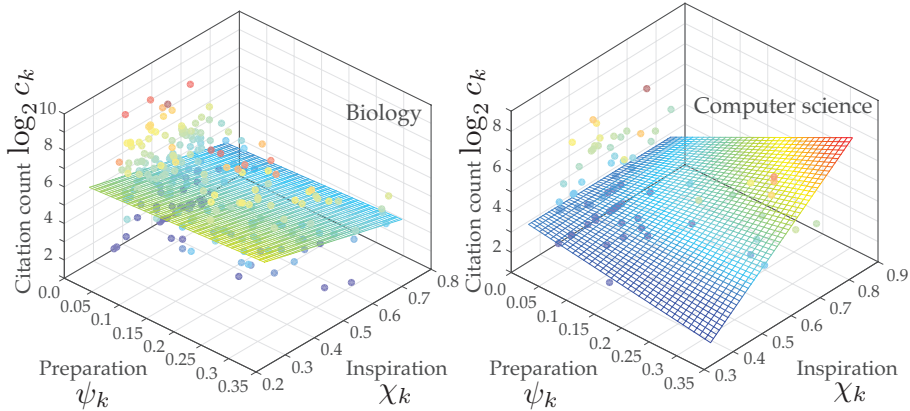}
\caption{Relationships of preparation and inspiration with citation counts. \label{fig:influence}}
\end{figure}

\section{Prediction of Enablers}
\label{sec:model}

In this section, levering the insights derived from our empirical study, we develop a predictive model that is able to identify the most promising enablers with respect to target papers. We then present efficient optimization algorithms for this model.

\subsection{Problem Formulation}

Without loss of generality, we start with the setting of a single target paper, and will discuss the extension to multiple target papers shortly. Let $k$ denote the target paper, which connects a set of disparate knowledge, as represented by a set of references ${\cal C}_k$. Among all existing literature ${\cal S}$, we intend to find a minimum set of reading papers ${\cal A} \subseteq {\cal S}$ that maximally facilitate to connect ${\cal C}_k$.

More formally, our goal is to solve the following optimization problem:
\begin{equation}
\label{eq:max}
\max_{{\cal A} \subseteq {\cal S}} \ell(\{\varphi_{i,j}\}_{i,j \in {\cal C}_k}) -\ell(\{\varphi_{i,j} - \Delta^{i,j}_{{\cal A}}\}_{i,j \in {\cal C}_k})
\end{equation}
with the constraint of $|{\cal A}|\leq \rho$ ($\rho$ as the threshold).

Let $R({\cal A}) \triangleq  \ell(\{\varphi_{i,j}\}_{i,j \in {\cal C}_k}) -\ell(\{\varphi_{i,j} - \Delta^{i,j}_{{\cal A}}\}_{i,j \in {\cal C}_k})$ and $c({\cal A}) \triangleq |{\cal A}|$. We can simplify this problem as:
\begin{displaymath}
\max_{{\cal A} \subseteq {\cal S}} R({\cal A}) \quad\textrm{subject to}\quad  c({\cal A})\leq \rho
\end{displaymath}

We may interpret $R({\cal A})$ and $c({\cal A})$ respectively as the ``reward'' and ``cost'' functions, and $\rho$ as the ``budget'' that one is allowed to spend. Next we discuss its optimization.

\subsection{Properties of Objective Function}

It is noticed that the objective function $R({\cal A})$ in Eqn.(\ref{eq:max}) possesses a set of interesting properties.
\begin{myitemize}
\item Non-negative - $R(\emptyset) = 0$; that is, $k$'s creativity does not change if no reading paper is taken into account.

\item Non-decreasing - If ${\cal A}\subseteq {\cal B} \subseteq {\cal S}$, then $R({\cal A})\leq R({\cal B})$. Intuitively, reading more papers can only reduce $k$'s creativity.

\item Diminishing return - Given ${\cal A}\subseteq {\cal B} \subseteq {\cal S}$, reading one more paper in addition to ${\cal A}$ improves $R(\cdot)$ at least as much as reading it in addition to ${\cal B}$. More formally,
\end{myitemize}

\begin{theorem}
Given ${\cal A}\subseteq {\cal B} \subseteq {\cal S}$ and a paper $x^* \in {\cal S}\setminus {\cal B}$, it holds that
\begin{displaymath}
  R({\cal A} \cup \{x^*\}) - R({\cal A}) \geq R({\cal B} \cup \{x^*\}) - R({\cal B})
\end{displaymath}
Thus, $R(\cdot)$ is a submodular function~\cite{Nemhauser:1978}.
\end{theorem}

\begin{proof}

First note that in $R(\cdot)$, for a given target paper $k$, $\ell(\{\varphi_{i,j}\}_{i,j \in {\cal C}_k})$ is input-independent. We thus focus our discussion on $\ell(\{\varphi_{i,j} - \Delta^{i,j}_{\cdot}\}_{i,j \in {\cal C}_k})$. Let us consider ${\cal A}\subseteq {\cal B} \subseteq {\cal S}$ and a reading paper $x^* \in {\cal S}\setminus {\cal B}$.

For ease of exposition, we assume that ${\cal A}$ (and ${\cal B}$) contains a special element $\oslash$, which corresponds to the case that no paper is selected from ${\cal A}$ (and ${\cal B}$). For each pair $i,j \in {\cal C}_k$, $\Delta_{\cal A}^{i,j} = \max_{x \in {\cal A}} \Delta^{i,j}_x$ and
$\Delta_{\cal B}^{i,j} = \max_{x \in {\cal B}} \Delta^{i,j}_x$. Since ${\cal A}\subseteq {\cal B}$, it holds that $\Delta_{\cal A}^{i,j} \leq \Delta_{\cal B}^{i,j}$. Let $x'$ = $\arg\max_{x \in {\cal B} \cup \{x^*\}} \Delta^{i,j}_{x}$. We consider the following three cases.

(i) $x^* = x'$. It is clear that $\Delta^{i,j}_{{\cal A} \cup \{x^*\}}$
= $\Delta^{i,j}_{{\cal B} \cup \{x^*\}}$ = $\Delta^{i,j}_{x^*}$; thus $\Delta^{i,j}_{{\cal A} \cup \{x^*\}} - \Delta^{i,j}_{{\cal A}}$ $\geq$ $\Delta^{i,j}_{{\cal B} \cup \{x^*\}} - \Delta^{i,j}_{{\cal B}}$.

(ii) $x^* \neq x'$ and $x' \in {\cal A}$. This implies $\Delta^{i,j}_{{\cal B} \cup \{x^*\}}$ = $\Delta^{i,j}_{{\cal A} \cup \{x^*\}}$ = $\Delta^{i,j}_{{\cal B}}$ = $\Delta^{i,j}_{{\cal A}}$. Thus, $\Delta^{i,j}_{{\cal A} \cup \{x^*\}} - \Delta^{i,j}_{{\cal A}}$ $=$ $\Delta^{i,j}_{{\cal B} \cup \{x^*\}} - \Delta^{i,j}_{{\cal B}}$.

(iii) $x^* \neq x'$ and $x' \not\in {\cal A}$. This implies $\Delta^{i,j}_{{\cal B} \cup \{x^*\}}$ = $\Delta^{i,j}_{{\cal B}}$ and $\Delta^{i,j}_{{\cal A} \cup \{x^*\}} \geq \Delta^{i,j}_{{\cal A}}$, i.e., $\Delta^{i,j}_{{\cal A} \cup \{x^*\}} - \Delta^{i,j}_{{\cal A}}$ $\geq$ $\Delta^{i,j}_{{\cal B} \cup \{x^*\}} - \Delta^{i,j}_{{\cal B}}$.

In all three cases above, it holds that $\Delta^{i,j}_{{\cal A} \cup \{x^*\}} - \Delta^{i,j}_{{\cal A}}$ $\geq$ $\Delta^{i,j}_{{\cal B} \cup \{x^*\}} - \Delta^{i,j}_{{\cal B}}$. Because the aggregation function $\ell(\cdot)$ (e.g., average, percentile, maximum) is non-decreasing (cf~\myref{sec:measurement}), it leads to $R({\cal A} \cup \{x^*\}) - R({\cal A}) \geq R({\cal B} \cup \{x^*\}) - R({\cal B})$.
\end{proof}

\begin{algorithm}
  \SetAlgorithmName{Algorithm}{List of operations}
  \SetAlgoLined
  \SetSideCommentRight
  \SetNoFillComment
  \KwIn{target paper $k$, existing literature ${\cal S}$, budget $\rho$}
  \KwOut{minimum set of enablers ${\cal A}$}

  ${\cal A} \leftarrow \emptyset$\;
  \For{$s = 1,\ldots,\rho$}{
  \tcc{\small greedy approach to find the next enabler}
  \For{each $x \in  {\cal S} \setminus {\cal A}$}{
    compute $R({\cal A} \cup \{x\}) - R({\cal A})$\;
  }
  \While{true}{
      find $x^* = \argmax_{x \in {\cal S} \setminus {\cal A}} R({\cal A} \cup \{x\}) - R({\cal A})$\;
      \tcc{\small temporal dynamics of information production-consumption dependency}
      pick $x^*$ with probability $m(\tau_{x^*,k}^p)$\;
      \lIf{$x^*$ is picked}{break}
      \lElse{remove $x^*$ from ${\cal S}$}
    }
  \lIf{$R({\cal A} \cup \{x^*\}) - R({\cal A}) = 0$}{break}
  ${\cal A} = {\cal A} \cup \{x^*\}$\;
  }
  return ${\cal A}$\;
  \caption{Finding minimum set of enablers \label{alg:sub}}
\end{algorithm}

\subsection{Optimization Algorithm}

In general, maximizing a submodular function is known to be NP-hard~\cite{Nemhauser:1978}. Yet, as each paper $x \in {\cal S}$ has unit cost, a greedy algorithm is applicable, which provides near-optimal guarantee for the found results. Next we introduce such an algorithm, as sketched in Algorithm~\ref{alg:sub}.

Let ${\cal A}_s$ denote the set of selected papers after the $s$-th step. Starting with an empty set ${\cal A}_0 = \emptyset$, and iteratively, at the $s$-th step, it adds the paper $x^* \in {\cal S}$ to ${\cal A}_{s-1}$, such that the following marginal gain is maximized:
\begin{displaymath}
x^* = \argmax_{x \in {\cal S} \setminus {\cal A}_{s-1}} R({\cal A}_{s-1} \cup \{x\}) - R({\cal A}_{s-1})
\end{displaymath}
This process stops if the marginal gain reaches zero or the budget is used up.

Algorithm~\ref{alg:sub} provides the following near-optimal guarantee for the found enablers.
\begin{theorem}[\cite{Nemhauser:1978}]
Given that the function $R(\cdot)$ is submodular, nondecreasing, and $R(\emptyset) = 0$, then the greedy algorithm finds ${\cal A}$, such that $R({\cal A}) \geq (1-1/e)\max_{|{\cal A}'|=\rho}R({\cal A}')$.
\end{theorem}

\subsection{Extensions}

Next we extend Algorithm~\ref{alg:sub} along two directions: (i) handling the case of multiple target papers and (ii) taking account of temporal dynamics of information production-consumption dependency.

\subsubsection*{Multiple Target Papers}
Let
$R(k, {\cal A}) = \ell(\{\varphi_{i,j}\}_{i,j \in {\cal C}_k}) -\ell(\{\varphi_{i,j} - \Delta^{i,j}_{{\cal A}}\}_{i,j \in {\cal C}_k})$. We update the objective function as:
\begin{displaymath}
R({\cal A}) = \sum_{k \in {\cal P}} \lambda_k R(k, {\cal A})
\end{displaymath}
where $\lambda_k$ is paper $k$'s weight, indicating its importance, with $\forall k \in {\cal P}, \lambda_k \geq 0$ and $\sum_{k \in {\cal P}}\lambda_k = 1$.

Under this multi-criterion setting, there may be cases that two solutions ${\cal A}, {\cal A}'$ are incompatible, i.e., $R(k, {\cal A}) > R(k, {\cal A}')$ while $R(k', {\cal A}) < R(k', {\cal A}')$. Instead of looking for the optimal solution, we resort to finding a Pareto-optimal solution~\cite{Boyd:2004:book}. A solution ${\cal A}$ is Pareto-optimal if no other solution ${\cal A}'$ satisfies that (i) $R(k,{\cal A}') \geq R(k,{\cal A})$ for all $k \in {\cal P}$ and $k \neq k'$, and (ii) $R(k',{\cal A}') > R(k',{\cal A})$ for a specific $k' \in {\cal P}$.

As $R(\cdot)$ is a non-negative linear combination of submodular functions, any solution that maximizes $R(\cdot)$ is guaranteed to be Pareto-optimal~\cite{Boyd:2004:book}. Moreover, since the submodularity is closed under the non-negative linear combinations, Algorithm~\ref{alg:sub} can be readily applied.

\subsubsection*{Temporal Dynamics}

Furthermore,  we take into account the temporal dynamics of information production-consumption dependency (cf.~\myref{sec:theory}). To build the temporal dynamics model, for $x \in {\cal C}_k \cap {\cal Q}$, we examine $x$'s publishing time ($t^p_x$) and $k$'s publishing time ($t^p_k$). The difference ($\tau_{x,k}^p = t^p_k - t^p_x$) reflects the temporal gap between $x$ and $k$'s publishing. Figure~\ref{fig:correlation}(b) shows that this interval fits a log-log distribution, denoted by $Pr(\tau)$. By discretizing and normalizing $Pr(\tau)$, we compute the probability that a paper published $\Delta t$ years ago influences a current paper, i.e., $m(\Delta t) = \int_{\tau = \Delta t}^{+\infty} Pr(\tau) \textrm{d}\tau$.
Similar to Algorithm~\ref{alg:sampling}, we incorporate this temporal dynamics into Algorithm~\ref{alg:sub} (line 7-9). We omit the details due to space limitations.

\subsection{Empirical Study}

\begin{figure}
\centering
\epsfig{width = 120mm, file = 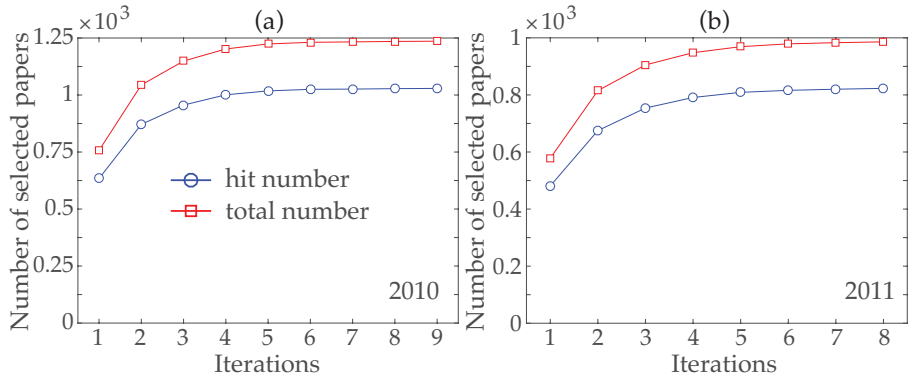}
\caption{Effectiveness of predicting the most influential enablers for the papers published in (a) 2010 and (b) 2011. \label{fig:pick}}
\end{figure}

To validate the effectiveness of our solution, we apply Algorithm~\ref{alg:sub} to predicting the most influential enablers for  papers published by Indiana University in a specific year $t$.
Specifically, we assume that each paper $k \in {\cal P}_t$ is of uniform importance, i.e., $\lambda_k = 1/|{\cal P}|$. We consider all papers published before year $t$ in the \scs{Mag} dataset as the literature space ${\cal S}$ (about 20 million papers for 2010 and 2011). We then measure the precision of our prediction algorithm: the proportion of selected papers ${\cal A}$ that appear in the reading set ${\cal Q}$ of the \scs{Click} dataset, i.e., $|{\cal A} \cap {\cal Q}|/|{\cal A}|$.

Figure~\ref{fig:pick} illustrates the measurement results for $t$ = 2010 and 2011. It is observed that our prediction model is fairly scalable: in both cases, it quickly converges using less than 10 iterations. Furthermore, among the selected papers ${\cal A}$, most of them indeed appear in the reading set ${\cal Q}$ (with precision over 83\% in both cases), highlighting the effectiveness of our predictive model.

\section{Conclusion \& Discussion}
\label{sec:conclusion}

In this work, we conducted a systematic study on creativity in scientific enterprise. For the first time, by directly correlating authors' raw information consumption behaviors with their knowledge products, we found remarkable predictability in scientific creative processes: of over 59.0\% papers across all scientific fields, 25.7\% of their creativity can be readily explained by information consumed by their authors. Leveraging these findings, we proposed a predictive framework that captures the impact of authors' information consumption over their future knowledge products. By using two web-scale, longitudinal real datasets, we demonstrated the efficacy of our framework in identifying the most critical knowledge to fostering target scientific innovations. Our framework is not limited to scientific creative processes only. Indeed, its mechanistic nature makes it potentially applicable for describing creative processes in other domains as well, such as musical, artistic and linguistics creativity.

This work also opens up several directions that are worth future investigations. For example, due to privacy and technology constraints, our study tracks information consumption and knowledge production at an organizational level. Thus, extending such study to an individual level could be fruitful and potentially shed new light on the nature of creativity. Furthermore, recent work has shown that various semantic features (e.g., author, content, venue) can be used to predict long-term impact of scientific artifacts in their early stages. Such semantic features could be integrated into our framework to train microscopic (author-, content-, and venue-specific) creativity models. Lastly, our model makes falsifiable prediction for creative processes, making it a viable candidate to assess and guide experimental studies, results of which can feed back to and improve the model with more accurate and realistic predictions.


\end{document}